\documentclass[conference]{IEEEtran}
\IEEEoverridecommandlockouts
\newcommand{\tsp}
{\mathsf{\scriptstyle{T}}}
\usepackage{verbatim} 
\usepackage{balance}

\newcommand{\R}{\mathbb{R}}

\newcommand{\E}{\mathbb{E}}
\usepackage{amsthm} 

\theoremstyle{definition}
\newtheorem{lma}{Lemma}

\theoremstyle{plain}
\newtheorem{thm}{Theorem}
\newcommand{\norm}[1]{\left\lVert#1\right\rVert}

\usepackage{array}
\usepackage{makecell}

\usepackage{blindtext}
\usepackage{xcolor}
\usepackage{cite}
\usepackage{amsmath,amssymb,amsfonts}
\usepackage{algorithmic}
\usepackage{graphicx}
\usepackage{textcomp}
\usepackage{xcolor,soul}
\def\BibTeX{{\rm B\kern-.05em{\sc i\kern-.025em b}\kern-.08em
    T\kern-.1667em\lower.7ex\hbox{E}\kern-.125emX}}

\begin{document}

\title{On the Capacity of Zero-Drift First Arrival Position Channels in Diffusive Molecular Communication
}

\author{\IEEEauthorblockN{Yen-Chi Lee}
\IEEEauthorblockA{\textit{Hon Hai (Foxconn) Research Institute} \\
Taipei, Taiwan \\
yen-chi.lee@foxconn.com}
\and
\IEEEauthorblockN{Min-Hsiu Hsieh}
\IEEEauthorblockA{\textit{Hon Hai (Foxconn) Research Institute} \\
Taipei, Taiwan \\
min-hsiu.hsieh@foxconn.com}
}

\maketitle

\begin{abstract}

\color{black}
Recent advancements in understanding the impulse response of the first arrival position (FAP) channel in molecular communication (MC) have illuminated its Shannon capacity. While Lee et al. shed light on FAP channel capacities with vertical drifts, the zero-drift scenario remains a conundrum, primarily due to the challenges associated with the heavy-tailed Cauchy distributions whose first and second moments do not exist, rendering traditional mutual information constraints ineffective. This paper unveils a novel characterization of the zero-drift FAP channel capacity for both 2D and 3D. Interestingly, our results reveal a 3D FAP channel capacity that is double its 2D counterpart, hinting at a capacity increase with spatial dimension growth. Furthermore, our approach, which incorporates a modified logarithmic constraint and an output signal constraint, offers a simplified and more intuitive formula (similar to the well-known Gaussian case) for estimating FAP channel capacity.
\color{black}

\end{abstract}


\section{Introduction}
\color{black}
Molecular communication (MC) represents a novel communication paradigm rooted in the transmission of information through molecular exchange \cite{nakano2013molecular,yeh2012new}. Owing to its nano-scale applicability and biocompatibility, MC emerges as a potent communication technique for nano-networks \cite{akyildiz2008nanonetworks,farsad2016comprehensive}. Within these systems, tiny message molecules (MM) act as the principal information carriers. A propagation mechanism is necessary for transporting MMs to the receiver (Rx), and this mechanism can be either diffusive (a.k.a. diffusion-based) \cite{pierobon2012capacity}, flow-based \cite{kadloor2012molecular}, or due to an engineered transport system like molecular motors \cite{moore2006design,gregori2010new}. Among these different propagation mechanisms, diffusive MC (sometimes in combination with a drift field) has been the most prevalent approach for MC theoretical research \cite{jamali2019channel}.
\color{black}

\color{black}
The molecular communication (MC) receiver's reception mechanism can be bifurcated into two primary categories, as delineated by \cite{jamali2019channel}: i) passive reception and ii) active reception. Within the realm of active reception, this study considers the widely recognized fully-absorbing Rx \cite{yilmaz2014three}. It is postulated that the Rx possesses the capability to precisely ascertain both the time \cite{srinivas2012molecular} and the position \cite{lee2016distribution,pandey2018molecular,akdeniz2018molecular} at the initial contact of the MM with the Rx. To investigate the Shannon capacity \cite{shannon1948mathematical} for this novel category of position channels, we adopt a straightforward geometric configuration for the receiver, specifically, a spacious receiving plane \cite{andrews2009accurate}. It is worth noting that this spacious plane model can serve as an approximation, particularly when the transmission distance is significantly shorter than the dimensions of the receiver, as discussed in prior works \cite{lee2016distribution,lee2023char}. Note that MMs are removed upon their interaction with the absorptive receiving plane \cite{yilmaz2014three,pandey2018molecular}. For a visual representation of this system model, please refer to Fig.~1.
\color{black}

There are various physical attributes of MMs that
can convey information.
The concept of the First Arrival Position (FAP) as an information-carrying property in MC literature was first introduced in the paper by Lee et al. in 2016 \cite{lee2016distribution}. Subsequently, this idea was further developed, particularly in the context of two-dimensional (2D) scenarios, in a later study conducted by Pandey et al. in 2018 \cite{pandey2018molecular}.
\color{black}
While the majority of studies within the realm of MC primarily concentrate on the First Arrival Time (FAT)\footnote{Readers interested in gaining a more comprehensive understanding of FAT-type modulation and its associated channel characteristics are encouraged to refer to \cite{srinivas2012molecular,li2014capacity} for detailed information.}
for absorbing receivers, it is noteworthy to acknowledge at least two reasons for giving due consideration to FAP as a viable alternative.
\begin{itemize}
    \item[1)] For every independent channel use—specifically, for each transmission of a single MM—the FAT information is one-dimensional (1D). In contrast, the FAP-type modulation can carry information in higher dimensions, up to \(n-1\), in an \(n\)-dimensional space. Consequently, the capacity of FAP channels, for each individual channel use, can increase significantly in higher-dimensional scenarios. This potentially allows FAP channels to surpass FAT channels in terms of capacity. We will establish the capacity formulas for FAP channels in Section~\ref{sec:MR} and validate this assertion for 2D and 3D cases.
    \item[2)] The second consideration pertains to time efficiency. In scenarios involving multiple-MM transmission, MMs may arrive out of sequence due to the inherent unpredictability of the diffusion process, leading to cross-over effects as documented in \cite{eckford2008molecular,jamali2019channel}. To mitigate cross-over effects in the MC system, the transmission duration for each symbol must not be overly short.\footnote{A simple way to conceptualize this is to imagine the necessity for a \emph{guard interval} between successive timing symbols.} Given this, for applications where time efficiency is paramount, exploring FAP-type modulation as an alternative solution might be advisable.
\end{itemize}

\color{black}
In this paper, our primary attention centers on one-shot transmission \cite{murin2018optimal}, referring specifically to transmission using a singular MM. To delve into the channel capacity of either FAP or FAT channels, initiating with a quantitative delineation of the channel impulse response is imperative. The one-shot FAT channel can be succinctly characterized as a time-invariant additive channel \cite{srinivas2012molecular}:
$
    t_{\text{out}} = t_{\text{in}} + t_{\text{n}},
$
where $t_{\text{in}}$ represents the release time, $t_{\text{out}}$ denotes the arrival time, and $t_{\text{n}}$ is the random time delay resulting from propagation dynamics. As established in \cite{srinivas2012molecular}, the variable $t_{\text{n}}$ adheres to the inverse Gaussian distribution. This channel is thus recognized as the Additive Inverse Gaussian Noise (AIGN) channel within the MC field. Subsequent research in \cite{srinivas2012molecular,li2014capacity} has presented bounds on the capacity of the AIGN channel and has further characterized the capacity-achieving input time distribution.
\color{black}

\color{black}
For FAP channels operating within \(n\)-dimensional spaces, the one-shot channel model is elegantly defined through an additive vector form, as detailed in \cite{lee2022arrival}:
$
    \mathbf{x}_{\text{out}} = \mathbf{x}_{\text{in}} + \mathbf{x}_{\text{n}},
$
where \(\mathbf{x}_{\text{in}}\) signifies the releasing position, \(\mathbf{x}_{\text{out}}\) indicates the arriving position, and \(\mathbf{x}_{\text{n}}\) captures the random position deviation induced by propagation mechanisms. 
It is crucial to recognize that the vectors \(\mathbf{x}_{\text{in}}\), \(\mathbf{x}_{\text{out}}\), and \(\mathbf{x}_{\text{n}}\) reside within the Euclidean vector space \(\mathbb{R}^{n-1}\). This configuration arises since one dimension is earmarked for physical movement, effectively excluding it from bearing information in that particular dimension.
Though past studies have articulated the probability density function of \(\mathbf{x}_{\text{n}}\) \cite{lee2016distribution,pandey2018molecular,lee2022arrival}, the quest to comprehend the channel capacity of this distinctive position channel is ongoing. A recent exploration by Lee et al. \cite{lee2023char} has elucidated the FAP channel capacity in contexts characterized by a vertical drift toward Rx. However, the channel capacity remains elusive in zero-drift scenarios. This manuscript aims to address this lacuna, unveiling precise formulas for the zero-drift FAP channel capacity for 2D and 3D dimensions, and introducing a novel logarithmic constraint, as detailed in Sections~III and IV.
\color{black}

The remainder of this paper is structured as follows. 
Section~\ref{pre-A} provides a succinct review of mutual information and the Shannon channel capacity framework. Section~\ref{sec:CM} outlines the channel model under scrutiny and introduces the \(\alpha\)-power constraint.
Our core findings are presented in Section~\ref{sec:MR}. We draw our conclusions in Section~\ref{sec:conclude}.
\color{black}


\section{Preliminary} \label{pre-A}
\color{black}
This section revisits foundational concepts surrounding mutual information and channel capacity. Consider \(X\) as a random variable with corresponding probability density function (pdf) \(f\). The support of \(f\) is represented as \(\mathcal{X}\). The differential entropy \(h(X)\) can be expressed as
$
h(X) := \E_X[-\ln f(X)] = -\int_{\mathcal{X}}  [\ln f(x)] f(x)dx.
$
Leveraging the above entropy definition, the mutual information between the input and output of the additive channel \(Y = X + N\) is:
\begin{align}
\begin{split}
    I(X;Y) &= h(Y) - h(Y|X)
    = h(Y) - h(X+N|X) \\
    &= h(Y) - h(X|X) - h(N|X) \\
    &= h(Y) - h(N).
\end{split}
\end{align}
Here, the assumption is that noise \(N\) operates independently of signal \(X\).
Remember that the (Shannon) channel capacity is captured as
$
C = \sup_{f_X(x)} I(X;Y).
\label{eq:inpc}
$
The supremum is evaluated over all possible \(X\) distributions under a selected constraint. For point-to-point communication or single-input-single-output (SISO) scenarios, output constraints can be adopted interchangeably, resulting in
$
C = \sup_{f_Y(x)} I(X;Y).
\label{eq:outpc}
$

Constraint selection is intimately tied to the noise distribution. To enhance understanding, we'll first contemplate a straightforward instance prior to diving into the specifics of FAP channels. For the additive Gaussian channel, where noise follows a Gaussian distribution, the prevailing power measure embraces the second moment constraint:
$
    \E[X^2] = \int_{\mathcal{X}} x^2 f(x)dx \leq P.
$
The constant \(P \geq 0\) denotes the upper power limit for all input signals. Based on this constraint, the channel capacity can be expressed as:
$
C = \sup_{f_X(x) :\  \E[X^2]\leq P} I(X;Y).
$
Consulting a standard reference such as \cite{cover1999elements}, the supremum of \(I(X;Y)\) is unveiled as
$
\frac{1}{2}\ln\left(\frac{\sigma^2+P}{\sigma^2}\right),
$
with the capacity-achieving distribution being Gaussian with variance \(P\), i.e. \(X \sim \mathcal{N}(0,P)\).

For subsequent comparison (presented in Section~\ref{sec:MR}) purpose, we can set \(A^2:=\sigma^2+P\), yielding
\begin{align}
\begin{split}
C_{\text{G}}
=\frac{1}{2}\ln\left(\dfrac{A^2}{\sigma^2}\right)
=\ln\left(\dfrac{A}{\sigma}\right),
\end{split}
\end{align}
where the subscript \(\text{G}\) hints the Gaussian channel.
\color{black}

\section{Channel Model} \label{sec:CM}
\begin{figure}[!t]  
\centerline{\includegraphics[width=0.48\textwidth]{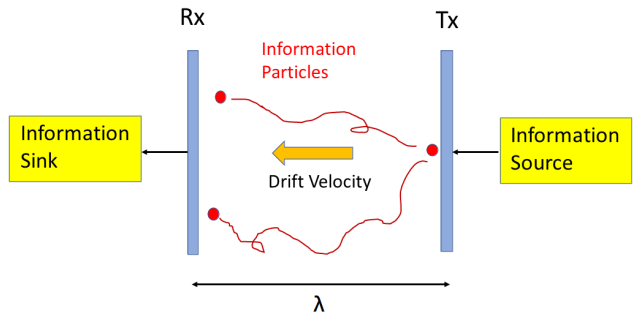}}
\caption{
This figure presents an abstract FAP channel model featuring hyperplane-shaped transmitters (Tx) and receivers (Rx). In this model, the transmitter is assumed to be transparent, permitting particles to traverse it without encountering any forces. It is important to observe that the emission point resides on the Tx hyperplane, and the Rx registers the arrival point when the information particle initially hits the Rx. Note that we consider \emph{zero drift} case for our main results.
}
\label{fig:2Da}
\end{figure}

We consider an additive vector channel:
$
    \mathbf{Y}=\mathbf{X}+\mathbf{N}
$
where the noise $\mathbf{N}$ is (multivariate) Cauchy distributed.\footnote{
We postpone the reason why FAP channels reduce to Cauchy channels under zero-drift assumption to Section~\ref{subsec:MR-A}.
}
Because the Cauchy distributions\footnote{
For a detailed characterization of Cauchy density functions, one can refer to Appendix~A.
} are heavy-tailed \cite{sigman1999primer}, traditional signal processing theory, which is tailored for finite second moment signals, does not apply directly to the Cauchy case.

Of the heavy-tailed probability models, the alpha-stable family \cite{pierce1997application} stands out, having demonstrated its prowess in simulating multiple access interference \cite{el2010alpha} and co-channel interference from a Poisson-distributed field of interferers \cite{gulati2010statistics}. In alpha-stable channels, the second moment is not a viable power measure due to its inherent infinity.

\color{black}
The task of discerning the channel capacity for alpha-stable additive noise—specifically when \(\alpha \geq 1\)—under an \(r\)-th absolute moment input constraint was adeptly addressed by \cite{fahs2012capacity} for symmetric alpha-stable noise scenarios. Their exploration led to the identification of an optimal capacity-achieving input that is both discrete and compactly defined. Meanwhile, a distinct concept of power, termed as the ``geometric power," was ushered in by \cite{gonzalez1997zero} to cater to heavy-tailed distributions \cite{sigman1999primer}. Within the context of this manuscript, we introduce an alternative power characterization schema for the \(\alpha\)-stable family, which is referred to as the \(\alpha\)-power, as detailed in \cite{fahs2017information}. It's pertinent to highlight that while the genesis of this framework traces back to \cite{fahs2017information}, our work repurposes its essence to adeptly navigate the challenges posed by zero-drift FAP channels in the realm of molecular communication.
\color{black}

Let us define a relative power measure $P_{\alpha}(\mathbf{X})$ satisfying the following:
\begin{itemize}
    \item[P1)] $P_\alpha(\mathbf{X}) \geq 0$, with equality if and only if $\mathbf{X}=0$ almost surely.
    \item[P2)] $P_\alpha(k\mathbf{X})=|k|P_\alpha(\mathbf{X})$, for any $k\in\R$.
\end{itemize}
For $\alpha$-stable distributions, it is known (see \cite{fahs2017information}) that
\begin{equation}
\begin{cases}
\E_\mathbf{X}\big[\ln(1+\norm{\mathbf{X}}^2)\big]<\infty,\quad &\alpha<2;\\
\E_\mathbf{X}\big[\norm{\mathbf{X}}^2\big] < \infty,\quad &\alpha=2.
\end{cases}
\end{equation}
We choose the power measure $P_1(\cdot)$ for multivariate Cauchy-like distributions as follows:
\begin{equation}
    \E_\mathbf{X}\left[\ln\left(1+\norm{\frac{\mathbf{X}}{P_1(\mathbf{X})}}^2 \right)\right]
    = w_2\left(\frac{1+p}{2};\frac{p}{2}\right),
\end{equation}
where the constant evaluation function $w_2(t;\alpha)$ can be written explicitly as
the difference of two digamma functions:
$
    w_2(t;a) = \psi(t) - \psi(t-a), \text{\ for\ } t>a.
$
Note that the values of $w_2$ function at $p=1$ and $p=2$ are $2\ln(2)$ and $2\ln(e)$ respectively. 

\color{black}

\color{black}
\section{Main Results}\label{sec:MR}
In this section, we delve into the intricacies of the FAP communication channel. We commence by demonstrating in Section~\ref{subsec:MR-A} that as the drift velocity in the fluid medium approaches zero, the 2D and 3D FAP channels simplify to univariate and bivariate Cauchy distributions, respectively.

Proceeding to Sections~\ref{MR-B} and~\ref{MR-C}, we integrate the concept of the \(\alpha\)-power, a relative power measure documented in \cite{fahs2017information}, into MC discussions. We explicitly define the signal space requisite to depict the no-drift FAP channels. Employing our introduced logarithmic constraint on the output signal space, we ascertain the Shannon capacity of the FAP channel for both 2D and 3D configurations in a closed-form manner. Key insights regarding the FAP channel capacity are encapsulated in Theorem~\ref{thm:T1} and Theorem~\ref{thm:T2}. 

Notice that in the literature, while \cite{fahs2014cauchy} provided a capacity formula for the univariate Cauchy channel (not for MC purpose) under a logarithmic constraint applied to the input signal space, directly extending their findings to higher-dimensional scenarios is challenging due to intricate integration procedures, as evident in \cite[Appendix~I]{fahs2014cauchy}. One way to tackle this is replacing the input constraints with output constraints, as demonstrated in Sections~\ref{MR-B} and~\ref{MR-C}.
\color{black}

\subsection{FAP Channel Reduces to Cauchy Channel Under Zero Drift Condition}\label{subsec:MR-A}
\subsubsection{In 2D Space}
From \cite{lee2022arrival} we know that
the FAP density function in 2D space, with a drift velocity $\mathbf{v}=(v_1,v_2)$, can be written as:
\begin{align}
\begin{split}
f_{Y \mid X}(\mathbf{y}|\mathbf{x})=&\ \dfrac{|\mathbf{v}|\lambda}{\sigma^2\pi}\exp\left\{\dfrac{-v_2 \lambda}{\sigma^2}\right\} \exp\left\{\dfrac{-v_1(x_1-y_1)}{\sigma^2}\right\}\\ &\cdot\dfrac{K_1\left(\dfrac{|\mathbf{v}|}{\sigma^2}\sqrt{(x_1-y_1)^2+\lambda^2}\right)}{\sqrt{(x_1-y_1)^2+\lambda^2}},
\end{split}
\label{eq:2DFAP}
\end{align}
where $\mathbf{x}=(x_1,\lambda)$, $\mathbf{y}=(y_1,0)$, and $\lambda$ is the \textit{transmission distance} between transmitter (Tx) and Rx. Note that in Eq.~\eqref{eq:2DFAP}, $\sigma^2$ is the microscopic diffusion coefficient which can be related to the macroscopic diffusion coefficient $D$ through the relation $\sigma^2=2D$, see \cite{farsad2016comprehensive}.
Note also that the special function $K_{1}(\cdot)$ in Eq.~\eqref{eq:2DFAP} is the modified Bessel function of the second kind (see \cite{bell2004special}) with order $\nu=1$.

Next we show that when drift velocity approaches zero, the 2D FAP density reduces to a univariate Cauchy distribution. 
We use a limit property about the special function $K_1(\cdot)$ from \cite{yang2017approximating}:
$
    \lim_{x\xrightarrow{}0} xK_1(x)=1.
$
Using this limit property, we can do the following calculation:
\begin{align}
\begin{split}
&\ f_{Y \mid X; \mathbf{v}=\mathbf{0}}(\mathbf{y}|\mathbf{x})\\
=&\ 
\lim_{|\mathbf{v}|\xrightarrow{}\mathbf{0}}
\frac{|\mathbf{v}|\lambda}{\sigma^2\pi}
\cdot\frac{K_1\left(\frac{|\mathbf{v}|}{\sigma^2}\sqrt{(x_1-y_1)^2+\lambda^2}\right)}{\sqrt{(x_1-y_1)^2+\lambda^2}}\\
=&\ 
\dfrac{\lambda}{\pi}
\lim_{|\mathbf{v}|\xrightarrow{}\mathbf{0}} \Bigg[
\frac{|\mathbf{v}|}{\sigma^2}\sqrt{(x_1-y_1)^2+\lambda^2}\\ &\quad\quad\quad\quad\quad \cdot\frac{K_1\left(\frac{|\mathbf{v}|}{\sigma^2}\sqrt{(x_1-y_1)^2+\lambda^2}\right)}{(x_1-y_1)^2+\lambda^2}
\Bigg]
\\
=&\ \dfrac{\lambda}{\pi}
\frac{1}{(x_1-y_1)^2+\lambda^2}.
\end{split}
\label{eq:2D-zero}
\end{align}
This calculation result can be interpreted as follows. Let $x_1$ be the input position and $y_1$ be the output position. The displacement of position is an additive noise $n$, following the density function:
$
\frac{\lambda}{\pi}
\frac{1}{n^2+\lambda^2},
$
which is Cauchy distributed. Please refer to Appendix~A for the form of density function of Cauchy random variables.

\subsubsection{In 3D Space}
From \cite{lee2022arrival,lee2016distribution} we know that
the FAP density function in 3D diffusion channel, with a drift velocity $\mathbf{v}=(v_1,v_2,v_3)$, can be written as:
\begin{equation}
\begin{aligned}
&f_{Y \mid X}(\mathbf{y}|\mathbf{x})\\
=\ & \frac{\lambda}{2 \pi} \exp \left\{-\dfrac{v_{3} \lambda}{\sigma^2}\right\} \exp \left\{\dfrac{v_{1}}{\sigma^2}\left(y_1-x_1\right)+\dfrac{v_{2}}{\sigma^2}\left(y_2-x_2\right)\right\} \\
& \cdot \exp \left\{-\dfrac{|\mathbf{v}|}{\sigma^2}\norm{\mathbf{y}-\mathbf{x}}\right\}\left[\frac{1+\frac{|\mathbf{v}|}{\sigma^2}\norm{\mathbf{y}-\mathbf{x}}}{\norm{\mathbf{y}-\mathbf{x}}^{3}}\right],
\end{aligned}
\label{eq:3DFAP}
\end{equation}
where $\sigma^2$ is the microscopic diffusion coefficient, the scale parameter $\lambda$ is the \textit{transmission distance} between Tx and Rx plane, $v_3$ is the longitudinal component of drift velocity in the direction parallel to the transmission direction, and $v_1$, $v_2$ are the components perpendicular (or transverse) to the transmission direction. Here we use the symbol $\norm{\cdot}$ to represent the Euclidean norm. Namely,
$
\norm{\mathbf{y}-\mathbf{x}}=\sqrt{(y_1-x_1)^2+(y_2-x_2)^2+\lambda^2},
$
where $\mathbf{x}=(x_1,x_2,\lambda)$, $\mathbf{y}=(y_1,y_2,0)$ are position vectors in $\R^3$.

Next we show that when drift velocity approaches zero, the 3D FAP density reduces to a bivariate Cauchy distribution.
Note that when $\mathbf{v}=\mathbf{0}$, all the exponential terms in Eq.~\eqref{eq:3DFAP} become $e^0=1$,
so that we have:
\begin{align}
\begin{split}
&\ f_{Y \mid X;\mathbf{v}=\mathbf{0}}(\mathbf{y}|\mathbf{x})\\
=&\ 
\frac{\lambda}{2 \pi}
\lim_{|\mathbf{v}|\xrightarrow{}\mathbf{0}}
\frac{1+\frac{|\mathbf{v}|}{\sigma^2}\norm{\mathbf{y}-\mathbf{x}}}{\norm{\mathbf{y}-\mathbf{x}}^{3}}
=\
\frac{\lambda}{2 \pi}
\frac{1}{\norm{\mathbf{y}-\mathbf{x}}^{3}}\\
=&\
\dfrac{\lambda}{2\pi}
\dfrac{1}{
\big((y_1-x_1)^2+(y_2-x_2)^2+\lambda^2\big)^{3/2}
}.
\end{split}
\end{align}
That is, if we regard the position channel as:
$
\mathbf{y}=\mathbf{x}+\mathbf{n},
$
then $\mathbf{n}$ follows a bivariate Cauchy distribution, see Eq.~\eqref{eq:simple-bi-ch} in Appendix~A for its pdf.
\color{black}


\subsection{Channel Capacity for 2D FAP Channel Under Logarithmic Constraint}\label{MR-B}
As mentioned in Section~\ref{subsec:MR-A}, a 2D FAP channel reduces to a Cauchy channel when the drift velocity approaches zero. Since we are considering a point-to-point (or SISO) communication scenario, the input constraint is actually equivalent to the output constraint. We choose to adopt output constraint to simplify the calculation.

It is known that a Cauchy distribution $X \sim {\rm Cauchy}(0,k)$ maximizes the entropy among all RVs\footnote{This is called the \textit{principle of maximum entropy} (see \cite{kapur1989maximum}) for Cauchy distribution.} that satisfy 
\begin{equation}
\E_X \ln \left[1+\left(\frac{X}{k}\right)^2\right] = 2\ln(2),
\end{equation}
or equivalently,
\begin{equation}
\E_X \ln \left[1+X^2\right] = 2\ln(1+k).
\end{equation}
The corresponding maximum entropy value is $\ln(4\pi k)$, see Appendix~A. Note that
the parameter $k$ restricts the dispersion of all feasible random variables $X$.


Consider a 2D FAP zero-drift channel (as shown in Section~\ref{subsec:MR-A}), the input-output relation can be written as
\begin{equation}
    Y=X+N,
\label{eq:29}
\end{equation}
where $N\sim {\rm Cauchy}(0,\lambda)$.
For continuous-variable channel capacity problem, we need to specify a family of distributions (i.e. feasible set) that are under consideration in order to prevent the capacity value from being infinite. Instead of writing down explicitly the constraint equality for $X$ or $Y$, we define
\begin{align}
\begin{split}
\mathcal{D}(A)
=\ 
&\Big\{\text{distributions\ }Y\ \Big|\ \exists k\in[\lambda,A]\\ 
&\quad\quad\quad\text{\ such that\ }
\E_Y \ln \Big[1+\Big(\frac{Y}{k}\Big)^2\Big] = 2\ln(2)
\Big\}.
\end{split}
\label{eq:domain-Y}
\end{align}
The parameter $A$ appeared in Eq.~\eqref{eq:domain-Y} indicates the ``largest allowed dispersion'' for the distributions in $\mathcal{D}(A)$. 
For later use, we also define
\begin{align}
\begin{split}
\mathcal{D}_k
=\ 
&\Big\{\text{distributions\ }Y\ \Big|\
\E_Y \ln \Big[1+\Big(\frac{Y}{k}\Big)^2\Big] = 2\ln(2)
\Big\},
\end{split}
\label{eq:domain-k}
\end{align}
so that we have
$
    \mathcal{D}(A)=\bigcup_{k\in[\lambda,A]} \mathcal{D}_k.
$
With these notations, now we can state our main results in the following two theorems.
\begin{thm}
The capacity of 2D FAP channel \eqref{eq:29} is
\begin{equation}
    C_{\rm{2D,FAP}}=C(A,\lambda)=\ln\left(\dfrac{A}{\lambda}\right)
\end{equation}
under the output logarithmic constraint
$
Y\in \mathcal{D}(A)
$
for some 
prescribed dispersion level $A$, where $A \geq \lambda$.
In addition, the corresponding capacity achieving output distribution is
$
Y^* \sim {\rm Cauchy}(0,A),  
$
or equivalently,
$
X^* \sim {\rm Cauchy}(0,A-\lambda).
$
\label{thm:T1}
\end{thm}

\color{black}
Given that the proof of Theorem~\ref{thm:T1} closely parallels that of Theorem~\ref{thm:T2}, we defer the proof details to Appendix~B and proceed directly to the examination of the 3D case.
\color{black}

\subsection{Channel Capacity for 3D FAP Channel Under Logarithmic Constraint}\label{MR-C}
\color{black}

It is known (see \cite{kapur1989maximum},\cite{kotz2004multivariate}) that a bivariate Cauchy distribution $\mathbf{X}\sim{\rm Cauchy}_2 (0,\Sigma=\text{diag}(k^2,k^2))$ maximizes the entropy among all bivariate RVs that satisfy 
\begin{equation}
\E_{\mathbf{X}} \ln \left[1+\norm{\frac{\mathbf{X}}{k}}^2\right] = 2\ln\left(e\right),
\end{equation}
The corresponding maximum entropy value is $\ln(2\pi e^3 k^2)$, see Appendix~A. Note that
the parameter $k$ restricts the dispersion of all feasible random vectors $\mathbf{X}$.

Consider a 3D FAP zero-drift channel (as shown in Section~\ref{subsec:MR-A}), the input-output relation can be written as
\begin{equation}
    \mathbf{Y}=\mathbf{X}+\mathbf{N},
\label{eq:45}
\end{equation}
where
$\mathbf{N}\sim{\rm Cauchy}_2 (\mathbf{0},\Sigma=\text{diag}(\lambda^2,\lambda^2))$.
We define
\begin{align}
\begin{split}
\mathcal{D}_{\text{bi}}(A)
=\ 
&\Bigg\{\text{distributions\ }\mathbf{Y}\ \Bigg|\ \exists k\in[\lambda,A]\\ 
&\quad\quad\quad\text{\ such that\ }
\E_\mathbf{Y} \ln \Big[1+\norm{\frac{\mathbf{Y}}{k}}^2\Big] = 2\ln(e)
\Bigg\}.
\end{split}
\label{eq:domain-Y-bi}
\end{align}
The parameter $A$ appeared in Eq.~\eqref{eq:domain-Y-bi} indicates the ``largest allowed dispersion'' for the distributions in $\mathcal{D}_{\text{bi}}(A)$. 
For later use, we also define
\begin{align}
\begin{split}
\mathcal{D}_{k,\text{bi}}
=\ 
&\Bigg\{\text{distributions\ }\mathbf{Y}\ \Bigg|\
\E_\mathbf{Y} \ln \Big[1+\norm{\frac{\mathbf{Y}}{k}}^2\Big] = 2\ln(e)
\Bigg\},
\end{split}
\label{eq:domain-k-bi}
\end{align}
so that we have
\begin{equation}
    \mathcal{D}_{\text{bi}}(A)=\bigcup_{k\in[\lambda,A]} \mathcal{D}_{k,\text{bi}}.
\end{equation}
The 3D FAP capacity formula is stated in the following theorem.
\begin{thm}
The capacity of 3D FAP channel \eqref{eq:45}  is
\begin{equation}
    C_{\rm{3D,FAP}}=C(A,\lambda)=2\ln\left(\dfrac{A}{\lambda}\right)
\end{equation}
under the output logarithmic constraint
$
\mathbf{Y}\in \mathcal{D}_{{\rm bi}}(A)
$
for some 
prescribed dispersion level $A$, where $A \geq \lambda$. In addition, the corresponding capacity achieving output distribution is
\begin{equation}
\mathbf{Y}^* \sim {\rm Cauchy}_2(0,{\rm diag}(A^2,A^2)).
\end{equation}
\label{thm:T2}
\end{thm}

\begin{proof}
For the additive channel we are considering, the mutual information can be written as
\begin{equation}
I(\mathbf{X};\mathbf{Y})=h(\mathbf{Y})-h(\mathbf{N}).
\end{equation}
The problem of finding the channel capacity turns out to be an optimization problem:
$
C(A,\lambda)=\sup_{\mathbf{Y}\in\mathcal{D}_{\text{bi}}(A)} I(\mathbf{X};\mathbf{Y}).
$
We have
\begin{align}
    C(A,\lambda)
    &=\sup_{\mathbf{Y} \in \mathcal{D}_{\text{bi}}(A)} \Big\{h(\mathbf{Y} ) - h(\mathbf{N} )\Big\}\\
    &=\Big\{\sup_{\mathbf{Y} \in \mathcal{D}_{\text{bi}}(A)} h(\mathbf{Y} )\Big\} - \ln(2\pi e^3\lambda^2), \label{eq:54}
\end{align}
where
\begin{align}
    \sup_{\mathbf{Y} \in \mathcal{D}_{\text{bi}}(A)} h(\mathbf{Y})
    &=
    \sup_{k\in[\lambda,A]}
    \Big\{ \sup_{\mathbf{Y}\in \mathcal{D}_{k,\text{bi}}} h(\mathbf{Y}) \Big\}\label{eq:55}\\
    &=
    \sup_{k\in[\lambda,A]}
    \ln(2\pi e^3 k^2) \label{eq:56}\\
    &= \ln(2\pi e^3 A^2). \label{eq:57}
\end{align}
Combining equations \eqref{eq:54} and \eqref{eq:57} yields
\begin{align}
\begin{split}
C(A,\lambda)
=&\ \ln(2\pi e^3 A^2) - \ln(2\pi e^3 \lambda^2)\\
=&\ \ln\left(\dfrac{A^2}{\lambda^2}\right)
=2\ln\left(\dfrac{A}{\lambda}\right).
\end{split}
\end{align}
Notice that the equalities \eqref{eq:55}-\eqref{eq:57} occur if and only if $\mathbf{Y}$ distributes as ${\rm Cauchy}_2(0,\text{diag}(A^2,A^2))$. Hence, we have proved Theorem~\ref{thm:T2}.
\end{proof}
\color{black}

We encapsulate the primary findings of this section in Table~I. Note that within this context, $A$ represents the output signal dispersion level specific to FAP channels. Moreover, by defining $A^2$ as $\sigma^2+P$, we establish a correspondence between the Gaussian case and the 2D FAP case. To visually illustrate these capacity curves, we present a graphical representation in Fig.~2.


\begin{table}[ht]
\caption{2D and 3D FAP channel capacity curves\\ with a comparison to Gaussian channels}
\label{tab:main}
\centering
\begin{tabular}{|r||c|c|c|}
\hline
\textbf{Channel Type} & \textbf{Gaussian} & \textbf{2D FAP} & \textbf{3D FAP} \\ 
\hline
Constraint & $L^2$ constr. & log constr. & log constr. \\ 
\hline
Capacity Formula & $\frac12\ln\left(\dfrac{\sigma^2+P}{\sigma^2}\right)$ & $\ln\left(\dfrac{A}{\lambda}\right)$ & $2\ln\left(\dfrac{A}{\lambda}\right)$\\ 
\hline
Parameter Range & N/A & $A\geq\lambda$ & $A\geq\lambda$\\
\hline
\end{tabular}
\label{tab:yourtable}
\end{table}

\begin{figure}[!t]  
\label{fig:capacity}
\centerline{\includegraphics[width=0.48\textwidth]{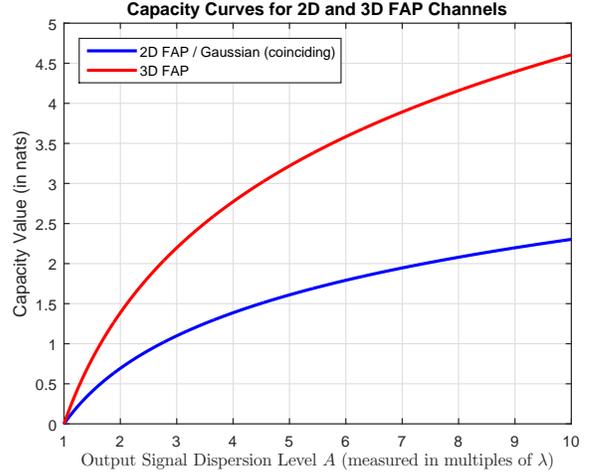}}
\caption{
Illustration depicting the capacity curves in relation to the output signal dispersion level $A$ for both 2D and 3D FAP channels. It is worth noting that the Gaussian curve coincides with the 2D FAP curve when we equate $A^2$ to $\sigma^2+P$ in the Gaussian case.
}
\label{}
\end{figure}

\section{Conclusion}\label{sec:conclude}
\color{black}
This study examines the capacity problem for zero-drift first arrival position (FAP) communication channel in diffusive molecular communication (MC) systems. 
Our findings reveal that when the drift velocity approaches zero, both the 2D and 3D FAP channels reduce to univariate and bivariate Cauchy distributions, respectively. 

Although a recent work \cite{lee2023char} investigate the FAP channel capacity under the case of positive \emph{vertical drift}, at the time of this writing, the determination of the Shannon capacity for the FAP channels with \emph{zero drift} (i.e. the multidimensional Cauchy channels) remains an open problem. This challenge stems from the heavy-tailed nature of the Cauchy distribution, which belongs to a broader family of $\alpha$-stable distributions that often lack a second moment. As a result, traditional energy constraints such as the variance of random variables are inadequate for characterizing the signal space.

To address this challenge, we introduce a relative power measure, the $\alpha$-power, into the field of molecular communication. We explicitly define the signal space required to characterize the zero-drift FAP channel capacity and derive the Shannon capacity for the FAP channel in 2D and 3D spaces in closed-form under a fresh logarithmic constraint. The capacity formulations articulated in Theorem~\ref{thm:T1} and Theorem~\ref{thm:T2} underscore an intriguing phenomenon: the 3D FAP channel's capacity eclipses its 2D counterpart by a factor of two when benchmarked against the identical dispersion parameter \(A\). This revelation amplifies the perspective that as spatial dimension \(n\) expands, FAP channels inherently harness greater informational capacity.
\color{black}

\bibliographystyle{IEEEtran}
\balance
\bibliography{main}

\appendices
\section{Density Function and Entropy of Multivariate Cauchy Distribution}
\label{appendix-A}
The Cauchy distribution, named after Augustin Cauchy, is a continuous probability distribution often used in statistics as a canonical example of distribution since both its expected value and its variance are undefined. It is also known, especially among physicists, as the Lorentz distribution.
In mathematics, it is closely related to the fundamental solution for the Laplace equation in the half-plane, and it is one of the few distributions that is stable (see \cite{pierce1997application,fahs2012capacity}) and has a density function that can be expressed analytically. (Other examples are normal distribution and L\'{e}vy distribution.)

The (univariate) Cauchy distribution has the probability density function (PDF) which can be expressed as:
\begin{equation}
    f(x;x_0,\gamma) =
    \dfrac{1}{\pi \gamma} 
    \dfrac{\gamma^2}
    {(x-x_0)^2+\gamma^2
    }
    =
    \dfrac{1}{\pi \gamma} 
    \dfrac{1}
    {1+(\frac{x-x_0}{\gamma})^2
    },
\label{eq:ch-dist}
\end{equation}
where $x_0\in \mathbb{R}$ is the location parameter, and $\gamma>0$ is the scaling parameter (see \cite{shao1993signal}). 
For the purpose of exploring the channel capacity, since the location parameter $x_0$ is irrelevant to the entropy, we may assume without loss of generality that $x_0=0$, yielding
the so called \textit{symmetry Cauchy} with PDF as:
\begin{equation}
    f(x;\gamma) =
    \dfrac{\gamma}{\pi} 
    \dfrac{1}
    {x^2+\gamma^2
    }
    =
    \dfrac{1}{\pi \gamma} 
    \dfrac{1}
    {1+(\frac{x}{\gamma})^2
    }.
\label{eq:sm-ch-dist}
\end{equation}
The entropy of Cauchy distribution can be evaluated by direct calculation. From Eq.~\eqref{eq:ch-dist}, we have
\begin{equation}
    h(X)=\ln(4\pi \gamma)
\end{equation}
whenever $X\sim \rm{Cauchy}(x_0,\gamma)$, and $x_0$ can be chosen arbitrarily.

As for multivariate Cauchy distribution, our notation system mainly follows \cite{kotz2004multivariate}.
Consider a $p$-dimensional Euclidean random vector $\mathbf{X}=(X_1,\cdots,X_p)^\tsp$ which follows multivariate Cauchy distribution.
We use the notation $\mathbf{X} \sim {\rm Cauchy}_p(\boldsymbol\mu,\Sigma)$ to specify the parameters, where $\boldsymbol\mu\in\R^{p\times 1}$ is the location vector, and $\Sigma\in\R^{p\times p}$ is the scale matrix describing the shape of the distribution.
The PDF of multivariate Cauchy distribution is given by
the following formula:
\begin{align} 
\label{eq:pvariate-ch-dist}
\begin{split}
    &f_{\mathbf{X}}^{(p)}(\mathbf{x};\boldsymbol\mu,\Sigma)\\
    =&\ 
    \dfrac{\Gamma(\frac{1+p}{2})}
    {\Gamma(\frac{1}{2})\pi^{\frac{p}{2}}
    |\Sigma|^{\frac{1}{2}}
    [1+(\mathbf{x}-\boldsymbol\mu)^\top\Sigma^{-1}(\mathbf{x}-\boldsymbol\mu)]^\frac{1+p}{2}
    }.
\end{split}
\end{align}
Note that $\Sigma$ is by nature a positive-definite square matrix.
For our later purpose, we mainly work in the case $p=2$. We provide the PDF of this special case here for convenience:
\begin{equation}
    f_{\mathbf{X}}^{(2)}(\mathbf{x};\boldsymbol\mu,\Sigma)=
    \dfrac{\Gamma(\frac{3}{2})}
    {\pi^{\frac{3}{2}}
    |\Sigma|^{\frac{1}{2}}
    [1+(\mathbf{x}-\boldsymbol\mu)^\top\Sigma^{-1}(\mathbf{x}-\boldsymbol\mu)]^\frac{3}{2}
    }.
\label{eq:bivariate-ch-dist}
\end{equation}
Eq.~\eqref{eq:bivariate-ch-dist} is the so-called bivariate Cauchy distribution.

Similar to the univariate case,
we can without loss of generality set $\boldsymbol\mu=\mathbf{0}$ in equations \eqref{eq:pvariate-ch-dist} and \eqref{eq:bivariate-ch-dist} for the purpose of entropy and channel capacity analysis.
When $\boldsymbol\mu = \mathbf{0}$, the multivariate Cauchy is called \textit{central}. For the case that the Cauchy is central, we abuse the notations ${\rm R}$ and $\Sigma$ to mean the same scale matrix.
We shall assume all Cauchy distributions considered in this paper to be central from now on. In addition, for the special case that $\Sigma$ is a diagonal matrix, 
namely
\begin{equation}
\Sigma=\text{diag}(\gamma^2,\gamma^2)
=\begin{bmatrix}
\gamma^2&0\\
0&\gamma^2
\end{bmatrix},
\end{equation}
the central bivariate Cauchy density function can be simplified into a more concise form:
\begin{equation}
\dfrac{1}{2\pi}
\dfrac{\gamma}{
(n_1^2+n_2^2+\gamma^2)^{3/2}
},
\label{eq:simple-bi-ch}
\end{equation}
where $\mathbf{n}=(n_1,n_2)$ is the Cauchy random vector.

Finally, the (differential) entropy of $p$-variate central Cauchy distribution can be found in [4]. We briefly state the results here for later use. Suppose the scale matrix ${\rm R}$ of a $p$-variate Cauchy $\mathbf{X}$ is prescribed, the entropy of $\mathbf{X}$ can be expressed as
\begin{equation}
    h(\mathbf{X};{\rm R})=\dfrac{1}{2}\ln|{\rm R}|+\Phi(p),
\end{equation}
where $|{\rm R}|$ represents the determinant of matrix ${\rm R}$; $\Phi(p)$ is a constant depending only on dimension $p$, and it can be evaluated through
\begin{align}
\begin{split}
    \Phi(p)
    =&\ \ln
    \left[
    \dfrac{\pi^{\frac{p}{2}}}{\Gamma(\frac{p}{2})}B\left(\frac{p}{2},\dfrac{1}{2}\right)
    \right]\\
    &\quad +\dfrac{1+p}{2}\left[
    \psi\left(\frac{1+p}{2}\right)-\psi\left(\frac12\right)
    \right].
\end{split}
\end{align}
In the above expression, $B(x,y)=\frac{\Gamma(x)\Gamma(y)}{\Gamma(x+y)}$ is the so-called \textit{beta function}, and
\begin{equation}
\psi(t):=\dfrac{d}{dt}\big[\ln\Gamma(t)\big]   
\end{equation}
is known as the \textit{digamma function}.

\section{Proof of Theorem~\ref{thm:T1}}
In this appendix section, we prove Theorem~\ref{thm:T1} step by step.
For the additive channel
\begin{equation}
    Y=X+N
\end{equation}
we are considering, the mutual information can be written as
$I(X;Y)=h(Y)-h(N)$ according to Section~\ref{pre-A}. The issue of finding the capacity of channel \eqref{eq:29} turns out to be an optimization problem:
\begin{align}
    C(A,\lambda)&=\sup_{Y\in\mathcal{D}(A)} I(X;Y).
\end{align}
The calculation is as follows. We have
\begin{align}
    C(A,\lambda)
    &=\sup_{Y \in \mathcal{D}(A)} \Big\{h(Y) - h(N)\Big\}\\
    &=\Big\{\sup_{Y \in \mathcal{D}(A)} h(Y)\Big\} - \ln(4\pi\lambda), \label{eq:39}
\end{align}
where
\begin{align}
    \sup_{Y \in \mathcal{D}(A)} h(Y)
    &=
    \sup_{k\in[\lambda,A]}
    \Big\{ \sup_{Y\in \mathcal{D}_k} h(Y) \Big\}\label{eq:40}\\
    &=
    \sup_{k\in[\lambda,A]}
    \ln(4\pi k) \label{eq:41}\\
    &= \ln(4\pi A). \label{eq:42}
\end{align}
Combining equations \eqref{eq:39} and \eqref{eq:42} yields
\begin{equation}
    C(A,\lambda) = 
    \ln(4\pi A) - \ln(4\pi \lambda) = \ln\left(\dfrac{A}{\lambda}\right).
\end{equation}
Notice that the equalities \eqref{eq:40}-\eqref{eq:42} hold if and only if $Y$ distributes as ${\rm Cauchy}(0,A)$. Hence, Theorem~\ref{thm:T1} is proved.
\label{appendix-B}

\section{Some Useful Facts about Cauchy Distribution}\label{pre-B}
\color{black}
In this appendix, we delve into specific properties of Cauchy distributions, building upon the foundational probability density function formula and entropy evaluations presented in Appendix~\ref{appendix-A}. 

The first two properties initially highlight the behavior of the independent sum of Cauchy distributions. To elucidate, the univariate symmetric Cauchy distribution remains invariant under independent addition, as articulated in the subsequent lemma.

\begin{lma}
Given two independent Cauchy random variables, $U \sim {\rm Cauchy}(0,\sigma)$ and $V \sim {\rm Cauchy}(0,\tau)$, their sum is described as:
\begin{equation}
Z = U + V \sim {\rm Cauchy}(0,\sigma+\tau).    
\end{equation}
\label{lma:L1}
\end{lma}
\noindent When considering the independent summation of bivariate Cauchy distributions, the lemma below encapsulates the relationship.

\begin{lma}
With $\mathbf{U} \sim {\rm Cauchy}_2\Big(\mathbf{0},\Sigma_1={\rm diag}(\sigma^2,\sigma^2)\Big)$ and $\mathbf{V} \sim {\rm Cauchy}_2\Big(\mathbf{0},\Sigma_2={\rm diag}(\tau^2,\tau^2)\Big)$ representing two independent bivariate Cauchy random vectors, their summation follows:
\begin{align}
\begin{split}
\mathbf{Z} &= \mathbf{U} + \mathbf{V} \\
&\sim {\rm Cauchy}_2\Big(\mathbf{0},\Sigma_3={\rm diag}\left((\sigma+\tau)^2,(\sigma+\tau)^2\right)\Big).
\end{split}
\end{align}
\label{lma:L2}
\end{lma}
\noindent Note that in Lemma~2, we employ the notation
\begin{equation}
{\rm diag}(a,b) =
\begin{bmatrix}
a & 0 \\
0 & b
\end{bmatrix}
\end{equation}
to succinctly represent diagonal matrices.

The final property examines the linear combination of components of a Cauchy vector, \(\mathbf{X}\). Given an arbitrary constant vector \(\mathbf{v} \in \R^{p \times 1}\), the lemma below delineates the relationship.

\begin{lma}
For a vector $\mathbf{X} \sim {\rm Cauchy}_p(\boldsymbol\mu,\Sigma)$, the relation is given by:
\begin{equation}
\mathbf{v}^\intercal \mathbf{X}
\sim {\rm Cauchy}(\mathbf{v}^\intercal \boldsymbol\mu,
\mathbf{v}^\intercal \Sigma \mathbf{v}),
\end{equation}
where \(\boldsymbol{\mu}\) signifies the location vector of \(\mathbf{X}\) and \(\Sigma\) represents its scale matrix.
\label{lma:L3}
\end{lma}
\color{black}

\end{document}